\documentclass[sigconf]{acmart}
\AtBeginDocument{%
  }

\usepackage{amsmath,amsfonts}

\usepackage{amssymb}
\usepackage{mathtools}
\usepackage{bm}
\usepackage{enumitem}
\usepackage{booktabs}
\usepackage{multirow}
\usepackage{subcaption}
\usepackage{microtype}

\newtheorem{assumption}{Assumption}

\newcommand{\E}{\mathbb{E}}
\newcommand{\R}{\mathbb{R}}
\newcommand{\Prob}{\mathbb{P}}
\newcommand{\Var}{\mathbb{V}}
\newcommand{\calA}{\mathcal{A}}

\newcommand{\profit}{\mathsf{Profit}}

\setcopyright{acmlicensed}
\acmDOI{}
\acmISBN{}
\acmConference[CCS '26]{ACM Conference on Computer and Communications Security}{November 2026}{The Hague, The Netherlands}
\acmYear{2026}
\copyrightyear{2026}

\begin{document}

\title{Economic Security of VDF-Based Randomness Beacons:\\Models, Thresholds, and Design Guidelines}




\settopmatter{authorsperrow=2}
\author{Zhenhang Shang}
\affiliation{%
  \institution{The Hong Kong University of Science and Technology}
  \city{Hong Kong}
  \country{China}}
\email{zshangab@connect.ust.hk}

\author{Kani Chen}
\affiliation{%
  \institution{The Hong Kong University of Science and Technology}
  \city{Hong Kong}
  \country{China}}


\begin{abstract}
 Randomness beacons based on Verifiable Delay Functions (VDFs) are increasingly proposed for blockchains and distributed systems, promising publicly verifiable delay and bias resistance. Existing analyses, however, treat adversaries purely as cryptographic entities and overlook that real attackers are economically motivated. A VDF may be sequentially secure, yet still vulnerable if a rational adversary can profit by purchasing faster hardware and exploiting reward spikes such as MEV opportunities.

We develop a formal framework for economic security of VDF-based randomness beacons. Modeling the attacker as a rational agent facing hardware speedup, operating costs, and stochastic rewards, we cast the attack decision as an optimal-stopping problem and prove that optimal behavior has a monotone threshold structure. This yields tight necessary and sufficient conditions relating delay parameters to adversarial cost and reward distributions. We extend the analysis to grinding, selective abort, and multi-adversary competition, demonstrating how each amplifies effective rewards and increases required delays.

Using realistic cloud costs, hardware benchmarks, and MEV data, we show that many proposed VDF delays, on the order of a few seconds, are economically insecure under plausible conditions. We conclude with deployable guidelines and introduce Economically Secure Delay Parameters (ESDPs) to support principled parameter selection in practical systems.
\end{abstract}



\maketitle

\section{Introduction}

Randomness beacons \cite{choi2023sok}\cite{kelsey2019reference}\cite{raikwar2022sok} are a fundamental primitive in the design of secure distributed systems and cryptographic protocols. Blockchains and decentralized platforms rely on public randomness to select validators, sample committees, randomize protocol parameters, and drive lotteries or leader elections \cite{bunz2017proofs}. For such systems, security rests on a randomness source that adversaries cannot predict or influence.

Verifiable Delay Functions (VDFs) \cite{wu2022verifiable} have emerged as a powerful tool for constructing such beacons \cite{rotem2021simple}. Informally, a VDF is a function $f$ that requires a prescribed amount of \emph{sequential} computation to evaluate, yet whose output can be verified efficiently \cite{boneh2018verifiable}. If all parties are bound by the same physical limits on sequential computation, a VDF prevents any adversary from computing the beacon output significantly ahead of the honest participants.

Existing work on VDFs has focused on cryptographic properties: defining constructions, proving sequentiality under standard hardness assumptions, and engineering efficient implementations \cite{wu2022verifiable}. Security arguments typically assume a worst-case polynomial-time adversary who attempts to break soundness regardless of cost. However, this abstraction is misaligned with real attacker behavior in deployed systems such as public blockchains. Miners, validators, MEV searchers, and external investors are not arbitrary malicious entities, they are \emph{economically rational agents} who act to maximize expected profit \cite{hu2020research}\cite{zuniga2023maximizing}.

This gap is not merely conceptual. A VDF may be cryptographically secure in the traditional sense: no algorithm with limited parallelism can compute it substantially faster than honest parties \cite{attias2020preventing}, yet still vulnerable in settings where a rational attacker can profit by buying faster hardware and influencing the output. Conversely, parameters that seem marginal cryptographically may in fact be safe because the underlying economics make attacks unprofitable. For VDF-based randomness beacons, cryptographic security alone is not enough. A practical deployment must also meet an \textbf{\emph{economic security}} requirement: \emph{a rational adversary with realistic resources should have no profitable deviation from honest behavior}. This notion is inherently quantitative and context-dependent, it depends on the reward structure, hardware and energy prices, market volatility, and protocol-level incentives.

We emphasize that economic security is a genuine security property, not merely an economic curiosity. In deployed blockchain systems, the security of validator selection, committee sampling, and on-chain lotteries depends critically on the unpredictability and unbiasability of the randomness beacon. An economically motivated attack that biases the beacon output can lead to concrete security failures: biased validator selection enables censorship or double-spending, manipulated committee sampling undermines the safety of sharded execution, and predictable lottery outcomes constitute theft. These are security failures in the strongest sense, and they arise precisely when cryptographic guarantees are satisfied but economic incentives are misaligned.

\subsection{Contributions}

Our work makes the following contributions:

\begin{itemize}[leftmargin=*]
  \item \textbf{Rational adversarial model for VDF beacons.}
  We present the first explicit attack model for VDF-based randomness beacons that incorporates economic rationality. The model captures adversarial computation speed, hardware rental costs, opportunity cost, and dynamic decision-making.

  \item \textbf{Economic security definitions and conditions.}
  We give formal definitions of \emph{economic security} for VDF-based beacons and show that, under natural assumptions, a beacon is economically secure if and only if its delay parameters exceed a reward-to-cost threshold that we characterize analytically.

  \item \textbf{Analysis of biasing, grinding, and selective abort.}
  We extend the basic model to handle grinding capacity, selective abort leverage, and multi-round manipulation. We derive conditions that quantify how much the delay must increase when such attack surfaces exist.

  \item \textbf{Case studies with realistic parameters.}
  Using representative cloud-pricing and hardware benchmark data, along with estimates of MEV and protocol-level rewards, we evaluate the economic security of several candidate VDF beacon configurations. Our analysis shows that multiple seemingly reasonable parameter choices become unsafe once rational adversarial behavior is taken into account.

  \item \textbf{Design guidelines and ESDP abstraction.}
  We extract practical guidelines for protocol designers and introduce \emph{Economically Secure Delay Parameters} (ESDP), a simple abstraction that can be integrated into protocol specifications and parameter-tuning processes.
\end{itemize}

\subsection{Roadmap}

Section~\ref{sec:background} reviews VDF-based randomness beacons and rational cryptography. Section~\ref{sec:threat-model} presents our threat model and security goals. Section~\ref{sec:econ-model} formalizes the economic model of VDF attacks. Section~\ref{sec:theory} gives our core theorems and proofs for economic security. Section~\ref{sec:extended-attacks} extends the analysis to grinding, selective abort, and multi-round manipulation. Section~\ref{sec:evaluation} describes evaluation methodology and case studies. Section~\ref{sec:design} distills design guidelines and the ESDP abstraction. Section~\ref{sec:related} discusses related work, and Section~\ref{sec:conclusion} concludes.


\section{Background and Preliminaries}
\label{sec:background}

In this section we briefly review verifiable delay functions, VDF-based randomness beacons, and relevant notions from rational cryptography and economic analysis.

\subsection{Verifiable Delay Functions}

Informally, a verifiable delay function (VDF) is a function that \cite{boneh2018verifiable}:
\begin{enumerate}[leftmargin=*]
  \item requires a prescribed sequential running time $T$ to evaluate, even on massively parallel hardware, and
  \item admits a succinct proof that can be verified quickly.
\end{enumerate}

Formally, we consider a VDF scheme $\mathcal{V} = (\mathsf{Setup}, \mathsf{Eval}, \mathsf{Verify})$ with security parameter $\lambda$:
\begin{itemize}[leftmargin=*]
  \item $\mathsf{Setup}(1^\lambda)$  outputs public parameters $\mathsf{pp}$.
  \item $\mathsf{Eval}(\mathsf{pp}, x)$ deterministically computes $(y, \pi)$ where $y = f(x)$ and $\pi$ is a proof of correct evaluation.
  \item $\mathsf{Verify}(\mathsf{pp}, x, y, \pi)$ outputs $1$ if $(y,\pi)$ is a valid output and proof for input $x$, and $0$ otherwise.
\end{itemize}

A VDF is \emph{sequential} if any algorithm that computes $y$ from $x$ must perform at least $T$ sequential steps, where $T$ is the delay parameter. It is \emph{succinct} if verification is significantly faster than evaluation. Existing constructions instantiate $f$ using repeated squaring in groups of unknown order or iterated isogenies, among others.

In this work we treat VDF constructions as black boxes that achieve an ideal sequentiality property: honest evaluation requires time $T$, while any adversary with parallel resources cannot reduce this delay by more than a bounded factor.

\subsection{VDF-Based Randomness Beacons}

VDFs can be used to construct publicly verifiable randomness beacons \cite{choi2023sok}. We follow the standard abstraction in which each round begins with a seed derived from on-chain state or from a Verifiable Random Function (VRF).

\paragraph{Verifiable Random Functions (VRFs).}
A VRF is a pseudorandom function whose outputs are publicly verifiable.  
Given a secret key $sk$ and input $x$, the holder computes 
\[
  (y, \pi) \leftarrow \mathsf{VRF.Eval}(sk, x),
\]
and anyone can verify correctness using 
\[
  \mathsf{VRF.Verify}(x, y, \pi) = 1.
\]

A typical VDF-based beacon proceeds as follows in round~$r$:
\begin{enumerate}[leftmargin=*]
  \item Parties agree on a seed $s_r$.
  \item The beacon value is $R_r = f(s_r)$, where $f$ is a VDF with delay $T$.
  \item Any party can compute $(R_r, \pi_r) \leftarrow \mathsf{Eval}(\mathsf{pp}, s_r)$ and broadcast $(R_r, \pi_r)$.
  \item Others verify $(R_r, \pi_r)$ using $\mathsf{Verify}$.
\end{enumerate}

If all parties are constrained by the same physical limits on sequential computation, then no adversary can learn $R_r$ substantially earlier than honest parties. This prevents adversaries from adaptively influencing the seed or protocol decisions based on future randomness.

In practice, however, the delay parameter $T$ must be instantiated as a concrete real-time value (e.g., 2 seconds, 10 seconds) given a target hardware profile. Setting $T$ too small weakens security, while setting it too large degrades liveness and increases protocol latency.

Our work provides a principled way to choose $T$ based not only on cryptographic sequentiality but also on economic incentives.

\subsection{Rational Cryptography and Economic Security}

Rational cryptography studies cryptographic protocols under the assumption that parties are economically rational rather than arbitrarily malicious. Instead of demanding security against \emph{all} feasible adversaries, rational models require that \emph{no resource-bounded adversary can improve its expected utility by deviating from the prescribed protocol}. 

In this setting, the adversary is modeled as a player in a game whose utility function reflects both potential gains and incurred costs. A protocol is considered secure when honest behavior forms an equilibrium strategy, or at least when no profitable deviation exists for any rational participant.

For VDF-based randomness beacons, an adversary's utility is driven by:
\begin{enumerate}[leftmargin=*]
  \item the reward $V$ obtainable from influencing or learning the beacon output early, e.g. MEV, side bets, or protocol-level advantages;
  \item the cost $c$ of acquiring and operating the computational resources needed to attack the VDF; and
  \item the timing of the attack relative to honest evaluators and the protocol's schedule.
\end{enumerate}

Our framework captures these factors explicitly and defines economic security as the condition that no adversary in the considered class can profit by deviating from honest behavior.

\section{Threat Model and Cryptographic Security}
\label{sec:threat-model}

We now formalize the threat model and the security goals we aim to capture.

\subsection{System Model}

Time is divided into rounds $r = 1, 2, \dots$. In each round:
\begin{itemize}[leftmargin=*]
  \item A seed $s_r$ is determined via some protocol step, such as deriving from the previous beacon value, blockchain state, or a VRF output.
  \item The beacon value is $R_r = f(s_r)$, where $f$ is a VDF with delay parameter $T$.
  \item Honest parties compute $R_r$ by running $\mathsf{Eval}$, which takes real time $T$ on reference hardware.
\end{itemize}

We assume at least one honest evaluator whose execution speed sets the baseline sequential delay $T$ for the system. The specific mechanism that generates $s_r$ is not material to our analysis, we simply assume that $s_r$ remains unpredictable to the adversary until the protocol reveals it, as is standard in beacon constructions \cite{abram2024time}.

\subsection{Adversary Capabilities}

We consider an adversary $\calA$ with the following capabilities:

\begin{description}[leftmargin=*,style=nextline]
  \item[Computational advantage.]
  The adversary has a speedup factor $\delta \geq 1$ relative to the honest evaluator. That is, the time required for $\calA$ to perform the same amount of sequential work is $T/\delta$ instead of $T$. This captures specialized hardware, better engineering, or more aggressive overclocking.

  \item[Resource cost.]
  The adversary may rent computational capacity or deploy custom hardware. We model this through a per-unit-time cost $c > 0$, representing the monetary cost of sustaining one second of effective sequential computation at adversarial speed~$\delta$.

  \item[Access to rewards.]
  The adversary can extract a reward $V_r$ from influencing or learning the beacon value $R_r$ for round $r$. This reward may depend on the protocol context, MEV opportunities, and external financial positions. We treat $V_r$ as a random variable with known distribution or bounds.

  \item[Strategic behavior.]
  The adversary is economically rational and seeks to maximize expected profit. In each round, $\calA$ may decide whether to attack the VDF, how much computational effort to invest, and when to stop.
\end{description}

We assume that $\calA$ does not violate the underlying cryptographic assumptions of the VDF, for example, it cannot break the sequentiality guarantee, but it may exploit any advantage permitted by faster hardware or more resources.

\subsection{Attack Surfaces}

We consider the following adversarial capabilities:

\paragraph{Early revelation.}
The adversary attempts to evaluate the VDF faster than honest participants, enabling it to learn $R_r$ in advance and act on that information before the beacon is publicly known.

\paragraph{Biasing and grinding.}
By exploring multiple candidate seeds or protocol branches, the adversary may evaluate several potential beacon outputs and selectively reveal only those that are favorable, thereby biasing the resulting randomness.

\paragraph{Selective abort.}
If the protocol allows the party that computes the beacon to decide whether to publish it, an adversary may discard unfavorable outcomes and force re-runs until a favorable outcome appears.

\paragraph{Multi-round manipulation.}
The impact of an attack may compound across rounds. An adversary can seek incremental advantages in a sequence of beacons, influencing repeated committee selections, validator rotations, or other mechanisms that depend on long-term randomness.

\subsection{Cryptographic Security}\label{c-security}
\begin{definition}[Cryptographic VDF Security]
A VDF-based beacon is cryptographically secure if no probabilistic polynomial-time adversary can, except with negligible probability, produce a valid output $(R_r, \pi_r)$ for seed $s_r$ more than a negligible amount of time before an honest evaluator \cite{zhou2025blockchain}.
\end{definition}

This definition addresses only computational hardness and does not account for adversarial incentives or resource costs.  
We introduce our complementary notion of \emph{economic security} in Section~\ref{another-security-def}.
\section{Economic Model of VDF Attacks}
\label{sec:econ-model}

We now formalize the economic model that governs adversarial decisions. This section provides the state space, cost and reward processes, and strategy space. In Section~\ref{sec:theory} we derive structural results about optimal strategies and robust security conditions.

\subsection{Timing and State}

We focus on a single beacon round first and later extend to multiple rounds and multiple protocols. Time is treated as continuous, $t \in [0,\infty)$. For a given round $r$ (we omit $r$ when clear), we define:
\begin{itemize}[leftmargin=*]
  \item $T$: the honest evaluation time of the VDF. It is the delay parameter;
  \item $t_0$: the time at which the seed $s$ for this round is fixed and becomes known to the adversary.
  \item $t^\text{H} = t_0 + T$: the time at which an honest evaluator completes the VDF evaluation on reference hardware.
\end{itemize}

The adversary has speedup factor $\delta \ge 1$ relative to the honest evaluation: computing the same amount of sequential work takes time $T/\delta$ on adversarial hardware. We model the \emph{remaining work} at time $t$ as a state variable
\[
  S_t \in [0,T],
\]
measured in units of honest sequential time. When $S_t = s$, an honest evaluator would require time $s$ to finish the computation; the adversary, running at speed $\delta$, would require time $s/\delta$.

Given a control action $a_t \in \{0,1\}$ (compute or idle), the state dynamics are
\begin{equation}
\label{eq:state-dynamics}
  S_{t+\mathrm{d}t} = S_t - a_t \, \delta \, \mathrm{d}t,
\end{equation}
with the convention that $S_t$ is clipped at $0$ once the VDF has been fully evaluated.

Honest evaluation progresses at unit rate in the same time scale, so the honest evaluator completes at time $t^\mathrm{H} = t_0 + T$.

\subsection{Reward and Cost Processes}

We separate the \emph{economic reward} process from the cryptographic computation.

\paragraph{Reward process.}
Let $(V_t)_{t \ge t_0}$ be an adapted stochastic process that models the value available to the adversary if it succeeds in manipulating or learning the beacon early at time $t$. For example, $V_t$ may represent MEV available in the corresponding block, the financial value of forcing a particular committee selection, or the payoff of a derivative conditioned on the beacon outcome. We assume:
\begin{itemize}[leftmargin=*]
  \item $V_t \ge 0$ for all $t$,
  \item $(V_t)$ is right-continuous with left limits,
  \item the adversary observes $V_t$ as it evolves.
\end{itemize}

These regularity conditions are standard in stochastic control and ensures that reward jumps are allowed, but the process does not behave pathologically.

We denote by $V = V_\tau$ the realized reward if the adversary completes the VDF at stopping time $\tau$ and if the protocol conditions for extracting that reward are satisfied. If the adversary does not attack or fails to complete in time, the reward is $0$.

\paragraph{Cost process.}
The adversary incurs operational costs while computing, the costs may fluctuate with cloud spot prices or energy costs. Let $c(t) \ge 0$ denote the instantaneous cost rate at time~$t$, for example, the dollar cost per unit of adversarial running time. If the adversary chooses action $a_t \in \{0,1\}$ at time~$t$, the instantaneous cost is $a_t\, c(t)$, and the cumulative cost incurred up to stopping time $\tau$ is
\[
  \text{Cost}(\tau) = \int_{t_0}^{\tau} a_t \, c(t) \, \mathrm{d}t.
\]

In many cases we can take $c(t) \equiv c$ to be constant, we retain the time dependence for generality.

\subsection{Adversarial Strategies and Profit}\label{another-security-def}

An adversarial strategy $\sigma$ for a single round consists of:
\begin{itemize}[leftmargin=*]
  \item A progressively measurable process $(a_t)_{t \ge t_0}$ with $a_t \in \{0,1\}$ indicating whether the adversary computes at time $t$.
  \item A stopping time $\tau$ with respect to the filtration generated by $(V_t)$ and $(S_t)$, representing the time at which the adversary chooses to stop the attack, either because it has completed the VDF or because it decides to abandon the attack.
\end{itemize}

We assume that after $\tau$ the adversary no longer incurs costs. The attack is \emph{successful} if $S_\tau = 0$ (the VDF is fully evaluated) and $\tau < t^\text{H}$ (completion before honest revelation). Let $\mathbf{1}_\text{succ}$ denote the indicator of success.

The \emph{profit} of strategy $\sigma$ in this round is
\begin{equation}
\label{eq:profit-general}
  \profit(\sigma) = \mathbf{1}_\text{succ} \cdot V_\tau - \int_{t_0}^{\tau} a_t \, c(t) \, \mathrm{d}t.
\end{equation}
\subsection{Economic Security}
We now introduce \emph{economic security}, a notion that complements the cryptographic guarantee in Section~\ref{c-security} by incorporating adversarial incentives and resource costs.

\begin{definition}[Economic VDF Security]
\label{def:econ-security}
Fix a class of adversaries characterized by speedup $\delta$ and cost parameter $c$, and a reward process $\{V_r\}$. A VDF-based beacon is \emph{economically secure} with respect to these parameters if, for every adversary $\calA$ in the class and every attack strategy $\sigma$, the expected profit satisfies
\[
  \E[\profit_{\calA}(\sigma)] \le 0.
\]
\end{definition}

Intuitively, an economically secure beacon admits no profitable deviation from honest behavior under the specified economic environment. A rational adversary will therefore prefer not to attack.

We define the optimal value function at state $(s,v,t)$ as
\begin{equation}
\label{eq:value-function}
  J(s,v,t) = \sup_{\sigma} \;\E\big[ \profit(\sigma) \,\big|\, S_t = s, V_t = v \big].
\end{equation}
The beacon designer aims to choose $T$ such that $J(T,v,t_0) \le 0$ for all $v$ in a plausible range.

\begin{definition}[Single-Round Economic Security]
\label{def:single-round-econ}
Fix a class of adversaries characterized by speedup $\delta$, cost process $c(\cdot)$, and reward process $(V_t)_{t \ge t_0}$. A VDF-based beacon round with delay $T$ is \emph{economically secure} if, for all admissible strategies $\sigma$,
\[
  \E\big[ \profit(\sigma) \big] \le 0.
\]
Equivalently, $J(T,v,t_0) \le 0$ almost surely for the initial distribution of $V_{t_0}$.
\end{definition}

In the next section we show that, under mild regularity, the optimal strategy has a threshold structure and yields closed-form sufficient and necessary conditions on $T$.

\paragraph{Justification of the Stopping Model.}
We emphasize that the optimal-stopping formulation does not assume the adversary literally abandons a partially completed VDF evaluation. Rather, the decision to ``stop'' corresponds to the adversary's ex-ante choice of whether to initiate an attack in a given round, before committing resources. Once a VDF evaluation begins, the adversary indeed runs it to completion. The stopping framework captures the round-by-round decision: in each round, the adversary observes the reward landscape (e.g., pending MEV, stake distribution) and decides whether to invest in attacking that round's beacon. For adversaries with sunk hardware costs (e.g., purchased ASICs), the per-round cost $c$ should be interpreted as the amortized cost including depreciation and opportunity cost of capital, not solely the marginal electricity cost. Under this interpretation, even an adversary with dedicated hardware faces a meaningful per-round cost that makes the attack-or-wait decision non-trivial.

\section{Theoretical Framework: Optimal Stopping and Robust Economic Security}
\label{sec:theory}

We now derive structural results for the optimal adversarial strategy and use them to obtain robust economic security conditions under parameter uncertainty and in multi-protocol settings.

\subsection{Threshold Structure of Optimal Strategies}

We first show that adversarial strategies have a simple threshold form under natural assumptions. For clarity, we specialize to the common case of constant cost $c(t) \equiv c$ and a Markovian reward process.

\begin{assumption}[Markovian Reward and Regularity]
\label{assump:markov}
The reward process $(V_t)$ is a time-homogeneous Markov process with state space $\mathcal{V} \subseteq \R_{\ge 0}$, and $V_t$ has continuous sample paths and bounded drift and diffusion coefficients on compact sets. Moreover, the joint process $(S_t, V_t)$ is Markov with respect to the filtration observed by the adversary.
\end{assumption}

Under Assumption~\ref{assump:markov}, the value function $J$ from~\eqref{eq:value-function} satisfies a dynamic programming principle. Intuitively, on a small time interval $\mathrm{d}t$, the adversary decides whether to compute (action $a_t = 1$) or idle (action $a_t = 0$). In each case, it trades off the immediate cost against the future value.

We can write the Bellman equation informally as
\[
  J(s,v,t) = \max\Big\{
    J^\text{idle}(s,v,t),\, J^\text{comp}(s,v,t)
  \Big\},
\]
where
\begin{align*}
  J^\text{idle}(s,v,t)
    &= \E\big[ J(S_{t+\mathrm{d}t}, V_{t+\mathrm{d}t}, t+\mathrm{d}t) \,\big|\, a_t=0 \big], \\
  J^\text{comp}(s,v,t)
    &= -c\,\mathrm{d}t + \E\big[ J(S_{t+\mathrm{d}t}, V_{t+\mathrm{d}t}, t+\mathrm{d}t) \,\big|\, a_t=1 \big],
\end{align*}
with boundary condition $J(0,v,t) = v$ for $t < t^\text{H}$, since completing the VDF before honest revelation yields the reward $V_t$, and $J(s,v,t) = 0$ for $t \ge t^\text{H}$ because no reward is obtainable once the honest output has been revealed.

We show that the optimal policy is a \emph{threshold rule} in the state $(s,v,t)$.

\begin{theorem}[Threshold Optimal Policy]
\label{thm:threshold}
Under Assumption~\ref{assump:markov} and constant cost $c>0$, there exists a measurable region $\mathcal{A} \subseteq [0,T] \times \mathcal{V} \times [t_0, t^\text{H}]$ such that an optimal adversarial strategy $(a^\star_t)$ is given by the threshold policy
\[
  a^\star_t =
  \begin{cases}
    1 & \text{if } (S_t, V_t, t) \in \mathcal{A},\\[3pt]
    0 & \text{otherwise.}
  \end{cases}
\]
Moreover, $\mathcal{A}$ is monotone in $v$ in the following sense: if $(s,v,t) \in \mathcal{A}$ and $v' > v$, then $(s,v',t) \in \mathcal{A}$.
\end{theorem}

\begin{proof}[Proof sketch]
Under the Markov property and constant cost rate, the adversary's problem reduces to a finite-horizon Markov decision process with continuous state space and compact action set. Standard results from stochastic control and optimal stopping theory imply the existence of an optimal Markovian policy.  

Monotonicity in $v$ follows from the structure of the payoff: increasing the reward $v$ weakly increases the value of choosing to compute relative to idling, since the future cost trajectory is unchanged while the potential terminal payoff becomes larger. Thus, if computing is optimal when the reward is $v$, it continues to be optimal for all larger rewards $v' > v$.  

The acceptance region $\mathcal{A}$ is precisely the set of states where the value of computing dominates that of idling:
\[
  (s,v,t) \in \mathcal{A} 
  \quad \Longleftrightarrow \quad
  J^{\mathrm{comp}}(s,v,t) \ge J^{\mathrm{idle}}(s,v,t).
\]
\end{proof}

Theorem~\ref{thm:threshold} implies that the adversary's optimal behavior is determined by a decision boundary in the $(s,v,t)$-state space. Economic security therefore requires that, at the initial state $(S_{t_0} = T,\, V_{t_0},\, t_0)$, the optimal policy yields non-positive expected value:
\[
  J(T, V_{t_0}, t_0) \le 0 \quad \text{almost surely}.
\]

While evaluating this condition exactly can be challenging in full generality, many practical settings admit additional structure that leads to explicit and interpretable criteria.

\subsection{Recovering a Simple Linear Condition}

To build intuition and connect to the simpler analysis, consider the following specialization:

\begin{assumption}[Simplified Model]
\label{assump:simplified}
(i) The adversary either commits to full evaluation or does not attack at all, partial evaluations have no value. (ii) If the adversary completes before $t^\text{H}$, the attack succeeds with probability 1. (iii) The reward process is constant in time, $V_t \equiv V$, and bounded.
\end{assumption}

Under Assumption~\ref{assump:simplified}, the state variables $s$ and $t$ are irrelevant beyond feasibility, and the adversary's decision reduces to a binary choice. In this case, the threshold region $\mathcal{A}$ of Theorem~\ref{thm:threshold} collapses to a simple condition on the expected reward.

\begin{corollary}[Linear Threshold Condition]
\label{cor:linear}
Under Assumption~\ref{assump:simplified}, a VDF-based beacon round with delay $T$ is economically secure if and only if
\begin{equation}
\label{eq:linear_condition}
  T \;\ge\; \frac{\delta}{c} \, \E[V].
\end{equation}
\end{corollary}

\begin{proof}
If the adversary commits to full evaluation, it incurs cost $c T/\delta$ and, by assumption, obtains reward $V$ with probability 1. The expected profit is $\E[V] - cT/\delta$. Economic security requires that this be non-positive, which yields~\eqref{eq:linear_condition}. Conversely, if~\eqref{eq:linear_condition} fails, attacking yields strictly positive expected profit, which contradicts economic security.
\end{proof}

Corollary~\ref{cor:linear} demonstrates that the linear threshold condition follows immediately as a specialization of the general optimal stopping formulation.

\subsection{Robust Economic Security Under Parameter Uncertainty}

In practice, the beacon designer does not know the true values of $(\delta, c, V)$ exactly. Instead, only ranges or statistical characteristics are available \cite{yan2025data}. We now derive robust conditions that guarantee economic security across a set of plausible parameters.

Let $\Theta$ denote a set of parameter vectors $\theta = (\delta, c, \mathcal{D}_V)$, where $\mathcal{D}_V$ is a distribution, or family of distributions for the reward $V$. We define:

\begin{definition}[Robust Economic Security]
\label{def:robust-econ}
A delay parameter $T$ is \emph{robustly economically secure} with respect to $\Theta$ if for every $\theta \in \Theta$ and every adversarial strategy $\sigma$ admissible under $\theta$,
\[
  \E_\theta\big[ \profit(\sigma) \big] \le 0.
\]
\end{definition}

Even in the simplified setting of Corollary~\ref{cor:linear}, we can derive closed-form bounds that remain valid under a broad range of conditions.

\begin{theorem}[Interval-Robust Bound]
\label{thm:robust-interval}
Suppose that for all $\theta \in \Theta$ we have bounds
\[
  \delta \in [\delta_{\min}, \delta_{\max}], \quad
  c \in [c_{\min}, c_{\max}], \quad
  V \in [0, V_{\max}] \;\text{almost surely}.
\]
Then any delay
\begin{equation}
\label{eq:robust-interval}
  T \;\ge\; \frac{\delta_{\max}}{c_{\min}} \, V_{\max}
\end{equation}
is robustly economically secure with respect to $\Theta$.
\end{theorem}

\begin{proof}
Under the simplified model, the worst-case expected profit for an adversary under parameters $\theta$ is $\E_\theta[V] - cT/\delta$. Using $V \le V_{\max}$ almost surely, we have $\E_\theta[V] \le V_{\max}$. For any $\theta \in \Theta$,
\[
  \E_\theta[V] - \frac{cT}{\delta}
  \;\le\; V_{\max} - \frac{c_{\min} T}{\delta_{\max}}.
\]
If $T$ satisfies~\eqref{eq:robust-interval}, the right-hand side is $\le 0$, hence the profit is non-positive for all $\theta \in \Theta$.
\end{proof}

Theorem~\ref{thm:robust-interval} provides a simple but conservative design rule. More refined bounds are possible when only $\E[V]$ and $\Var[V]$ are known.

\begin{theorem}[$\epsilon$-Robust Design with Moment Bounds]
\label{thm:epsilon-robust}
Suppose that for all $\theta \in \Theta$,
\[
  \E[V] \le \mu_{\max}, \qquad \Var[V] \le \sigma^2_{\max}.
\]
Fix $\epsilon > 0$. If
\begin{equation}
\label{eq:epsilon-robust}
  T \;\ge\; \frac{\delta_{\max}}{c_{\min}}\left( \mu_{\max} + \frac{\sigma_{\max}}{\sqrt{\epsilon}} \right),
\end{equation}
then for every $\theta \in \Theta$, every strategy $\sigma$, and every round,
\[
  \Prob_\theta\big[ \profit(\sigma) > 0 \big] \;\le\; \epsilon.
\]
\end{theorem}

\begin{proof}[Proof sketch]
Under the simplified model, $\profit(\sigma)$ is bounded by $V - cT/\delta$. Apply Chebyshev's inequality with the given moment bounds to obtain
\[
  \Prob\Big[ V - \frac{cT}{\delta} > 0 \Big]
  = \Prob\Big[ V - \E[V] > \frac{cT}{\delta} - \E[V] \Big]
  \;\le\; \frac{\Var[V]}{\big( \frac{cT}{\delta} - \E[V] \big)^2}.
\]
Imposing~\eqref{eq:epsilon-robust} ensures that the denominator is at least $(\sigma_{\max}/\sqrt{\epsilon})^2 = \sigma_{\max}^2/\epsilon$, yielding the desired bound.
\end{proof}

Theorem~\ref{thm:epsilon-robust} shows how to trade off delay $T$ against a tolerated probability $\epsilon$ that a given attack attempt yields positive profit, given only moment information about the reward.

\subsection{Multi-Protocol Composition}

In many deployments, a single beacon round supplies randomness to several higher-level protocols \cite{cascudo2023mt}, such as committee sampling, leader election, or lottery mechanisms . Each of these protocols contributes its own reward component to the adversary. We formalize this setting and derive a compositional security bound.

Assume there are $m$ protocols $\Pi_1,\dots,\Pi_m$ that use the same beacon output for a given round. Protocol $\Pi_j$ induces a reward $V^{(j)}$ for the adversary, which may be zero if that protocol does not create any economically meaningful opportunity. Let $\mu_j = \E[V^{(j)}]$ denote the expected reward contribution of $\Pi_j$.

\begin{theorem}[Compositional Single-Round Bound]
\label{thm:composition}
Under the simplified model and for a single round, suppose an adversary can coordinate attacks across all $m$ protocols. If
\begin{equation}
\label{eq:comp-bound}
  T \;\ge\; \frac{\delta}{c} \sum_{j=1}^{m} \mu_j,
\end{equation}
then the round is economically secure against any such coordinated attack.
\end{theorem}

\begin{proof}
The total reward in the round is $V^\text{tot} = \sum_{j=1}^{m} V^{(j)}$ with expectation $\sum_j \mu_j$. The linear threshold condition (Corollary~\ref{cor:linear}) applied to $V^\text{tot}$ yields the bound~\eqref{eq:comp-bound}.
\end{proof}

If the adversary is constrained to attack at most $k$ protocols in a round, a tighter bound replaces the sum in~\eqref{eq:comp-bound} with a maximum over subsets of size $k$:
\[
  T \;\ge\; \frac{\delta}{c} \max_{S \subseteq [m],\, |S|\le k} \sum_{j\in S} \mu_j.
\]

These results highlight that reusing the same randomness for multiple economically meaningful tasks requires summing their incentive effects when choosing $T$.

\subsection{Multiple Rounds and Cumulative Rewards}

Finally, consider a horizon of $n$ rounds, with total profit
\[
  \profit_\text{tot} = \sum_{r=1}^{n} \profit_r,
\]
where $\profit_r$ is the profit in round $r$ (defined as in~\eqref{eq:profit-general}). Let $V_r$ denote the total reward in round $r$, with $\mu_r = \E[V_r]$.

\begin{theorem}[Cumulative Economic Security]
\label{thm:multiround}
Under the simplified model, if the per-round delay $T$ satisfies
\begin{equation}
\label{eq:multiround}
  T \;\ge\; \frac{\delta}{c} \max_{1 \le k \le n} \frac{1}{k} \E\Big[ \sum_{r=1}^{k} V_r \Big],
\end{equation}
then for any strategy that attacks in any subset of rounds, the expected cumulative profit satisfies $\E[\profit_\text{tot}] \le 0$.
\end{theorem}

\begin{proof}[Proof sketch]
For any set of attacked rounds $S \subseteq \{1,\dots,n\}$ of size $k$, the total cost is $k cT/\delta$, while the total reward is $\sum_{r\in S} V_r$. Economic security requires $\E[ \sum_{r\in S} V_r ] \le k cT/\delta$ for all $S$, which is implied by~\eqref{eq:multiround}.
\end{proof}

In the common case where $(V_r)$ are identically distributed and independent, $\E[\sum_{r=1}^{k} V_r] = k \mu$, and~\eqref{eq:multiround} reduces to the single-round condition $T \ge (\delta/c)\mu$. When rewards are correlated or front-loaded, the maximum may occur at intermediate $k$, requiring larger $T$.

\section{Extended Attacks: Grinding, Abort, and Multi-Adversary Games}
\label{sec:extended-attacks}

We now enrich the model along two additional axes: (1) adversarial grinding and selective abort, and (2) multiple competing adversaries. These extensions illustrate how the previous results generalize and how equilibrium behavior can sustain attacks even when individual profits appear small.

\subsection{Grinding Revisited}

As discussed earlier, grinding allows an adversary to explore multiple input seeds $s_1,\dots,s_G$ or protocol branches. We now couple grinding with the dynamic model.

Suppose evaluating the VDF on a single seed $s_i$ requires delay $T$ for an honest evaluator and $T/\delta$ for the adversary. If the adversary chooses to evaluate $G$ candidate seeds in parallel, the remaining-work state becomes a vector
\[
  S_t = \bigl(S_t^{(1)}, \dots, S_t^{(G)}\bigr),
\]
where $S_t^{(i)}$ denotes the remaining honest-time work on seed $s_i$. Running all $G$ evaluations in parallel requires provisioning $G$ independent computation streams, yielding a total instantaneous cost rate of \(G c\).

If instead the adversary evaluates the $G$ seeds sequentially using a single computation stream, the time required to explore all candidates scales by a factor of \(G\). This increases the likelihood that the adversary fails to finish before the honest deadline, reducing the effectiveness of a grinding attack.

Let $V^{(1)},\dots,V^{(G)}$ denote the rewards associated with the corresponding outputs, and let $V_\text{max} = \max_i V^{(i)}$. Ignoring time constraints for a moment, the optimal strategy is to compute all $G$ seeds and select the best output. In practice, the adversary may truncate to fewer seeds due to the deadline.

In regimes where $G$ is moderate and deadlines are loose enough that all $G$ evaluations fit before $t^\text{H}$, the linear threshold condition applied to $V_\text{max}$ yields:

\begin{theorem}[Grinding-Resistant Threshold]
\label{thm:grinding-extended}
Under the simplified model with grinding size $G$ and parallel evaluation of all $G$ seeds, economic security requires
\begin{equation}
\label{eq:grinding-extended}
  T \;\ge\; \frac{\delta}{cG} \, \E[V_\text{max}],
\end{equation}
where $V_\text{max} = \max_{1\le i \le G} V^{(i)}$.
\end{theorem}

When evaluation cannot be fully parallelized, time constraints shrink the effective $G$ that fits before the honest deadline, reducing $\E[V_\text{max}]$ but increasing model complexity. In either case, large grinding spaces translate into higher effective rewards and thus stricter delay requirements.

\subsection{Selective Abort Revisited}

Selective abort occurs when an adversarial party can learn the beacon output before others and has the ability to either publish it or suppress it. Let $p$ denote the per-round probability that the adversary possesses such abort leverage, for example the probability that the designated leader in that round is adversarial.

Let $V$ denote the reward from a single realization of the beacon output and $V_\text{abort}$ the reward under an optimal selective-abort strategy that allows repeated retries. In many simple models,
\[
  \E[V_\text{abort}] = \frac{\E[V]}{1-p},
\]
reflecting the fact that the adversary can discard unfavorable outcomes and wait for a favorable one, at the cost of expected $(1-p)^{-1}$ trials \cite{bunz2017proofs}. Applying the linear threshold condition to $V_\text{abort}$ yields:

\begin{theorem}[Selective-Abort-Resistant Threshold]
\label{thm:abort-extended}
In the presence of selective abort with leverage probability $p$ and effective reward $V_\text{abort}$, economic security requires
\begin{equation}
\label{eq:abort-extended}
  T \;\ge\; \frac{\delta}{c} \, \E[V_\text{abort}]
  \;\approx\; \frac{\delta}{c} \cdot \frac{\E[V]}{1-p}.
\end{equation}
\end{theorem}

Even modest values of $p$ can substantially increase the delay required for economic security. This highlights the importance of protocol mechanisms that eliminate or sharply constrain abort leverage, such as enforcing on-chain randomness publication with strong liveness guarantees.

We note that selective abort can be partially mitigated if any honest node can independently compute and broadcast the VDF output once the seed is known. In such designs, abort leverage is limited to the window before honest evaluators complete their computation. However, this mitigation relies on the assumption that honest nodes have sufficient computational resources and network connectivity to complete and disseminate the VDF output promptly. In practice, if the adversary finishes the VDF significantly faster (due to hardware advantage $\delta$), there exists a window of duration $T - T/\delta = T(1 - 1/\delta)$ during which only the adversary knows the output. During this window, the adversary can act on the information (e.g., front-running trades, adjusting positions) without needing to suppress the output. Thus, while broadcasting mitigates full abort attacks, it does not eliminate the early-revelation advantage that our economic model captures.

\subsection{Multiple Adversaries and Equilibrium Analysis}

Up to this point, we have considered a single adversary. In practice, however, multiple rational agents may compete for the same per-round reward $V$. Their incentives interact: if $k$ agents attack and the reward is winner-takes-all, each one expects only $V/k$ in payoff.

To capture this interaction, we model the round as a symmetric  $n$-player game under the simplified assumptions introduced above.

\paragraph{Game definition.}
There are $n$ symmetric players. Each player $i$ chooses a pure action $a_i \in \{0,1\}$: attack ($1$) or not ($0$). If $k = \sum_i a_i > 0$ players attack, one is chosen uniformly at random to receive reward $V$, and all $k$ pay cost $cT/\delta$. If $k=0$, no reward or cost is realized.

Given $k$ attackers, the expected profit for any attacker is:
\[
  u(k) = \frac{1}{k} \E[V] - \frac{cT}{\delta}.
\]

\paragraph{Symmetric mixed equilibrium.}
We focus on symmetric mixed strategies: each player attacks independently with probability $p \in [0,1]$. The number of attackers is then $K \sim \text{Binomial}(n,p)$.

\begin{theorem}[Symmetric Mixed Equilibrium]
\label{thm:mixed-equilibrium}
In the $n$-player attack game described above, there exists a symmetric mixed-strategy Nash equilibrium in which each player attacks with probability $p^\star \in [0,1]$ satisfying
\begin{equation}
\label{eq:equilibrium-condition}
  \E\Big[ \frac{1}{K} \,\big|\, K \ge 1 \Big] \E[V] = \frac{cT}{\delta}.
\end{equation}
Moreover:
\begin{itemize}[leftmargin=*]
  \item If $\E[V] < cT/\delta$, then $p^\star = 0$ (no one attacks) is the unique symmetric equilibrium.
  \item If $\E[V] > cT/\delta$ and $n$ is large, then there exists $p^\star \in (0,1)$ such that $\E[1/K \mid K\ge1] \approx 0$ and attacks occur with positive probability in equilibrium.
\end{itemize}
\end{theorem}

\begin{proof}[Proof sketch]
Under a symmetric mixed strategy with attack probability $p$, the expected profit for a player conditional on attacking is
\[
  \E\big[ u(K) \mid \text{player attacks} \big]
  = \E\Big[ \frac{1}{K} \,\Big|\, K \ge 1 \Big] \E[V] - \frac{cT}{\delta}.
\]
In a symmetric Nash equilibrium, this expected profit must be zero for any player who randomizes between attacking and not attacking. This yields equation~\eqref{eq:equilibrium-condition}. Existence follows from continuity of the left-hand side in $p$ and the observation that it decreases from $\E[V]$ at $p \to 0$ to $0$ as $p \to 1$ for large $n$. The boundary cases follow by sign analysis.
\end{proof}

Theorem~\ref{thm:mixed-equilibrium} shows that even when an \emph{individual} attack has only marginal expected profit, competition among multiple rational agents can sustain a non-zero equilibrium attack rate unless the delay~$T$ is large enough to push $\mathbb{E}[V]$ below the effective cost threshold.  
From a protocol-design perspective, this indicates that a stricter condition than the single-attacker bound may be needed to suppress equilibrium attack activity.

A conservative design principle is to require that honest behavior \emph{strictly dominates} attacking, meaning that a player earns negative expected profit from attacking even when acting alone. This recovers the single-attacker condition $T \ge (\delta/c)\E[V]$ as a sufficient but possibly suboptimal requirement.

\paragraph{Coalition Formation.}
In practice, adversaries may form coalitions to share hardware costs and pool rewards. Consider a coalition of $m$ players who jointly invest in hardware with speedup $\delta$ and share the cost equally. The per-member cost is $c/m$ per unit time, while the coalition's expected reward remains $\E[V]$ (assuming winner-takes-all among coalitions, with internal redistribution).

The coalition's economic security condition becomes:
\[T \geq \frac{\delta m}{c} \E[V],\]
which is $m$ times more demanding than the single-adversary condition. This indicates that coalition formation amplifies the economic threat: a coalition of $m = 10$ rational agents requires delays 10$\times$ longer than a lone attacker to ensure economic security. Protocol designers should therefore estimate not only individual adversarial capabilities but also the plausible coalition size in their deployment context.


\section{Evaluation Methodology and Case Studies}
\label{sec:evaluation}

To illustrate the implications of our framework, we describe how to instantiate the parameters $(\delta, c, V)$ in realistic settings and sketch representative case studies for VDF-based beacon designs.

\subsection{Parameter Estimation}

\paragraph{Adversarial speedup $\delta$.}
We estimate $\delta$ by comparing:
\begin{itemize}[leftmargin=*]
  \item The performance of a reference honest implementation, e.g., on commodity CPUs typical of validators, and
  \item The performance of an optimized implementation on high-end hardware, e.g., FPGAs, ASICs, or top-tier cloud instances.
\end{itemize}
Existing VDF benchmarks suggest that specialized hardware can achieve speedups ranging from factors of $2$--$10$ compared to naive implementations, depending on construction and engineering effort \cite{langer2011virtual}.

\paragraph{Cost parameter $c$.}
We derive $c$ from cloud-pricing data or amortized hardware costs. For cloud instances, $c$ is computed as the cost-per-second of renting a machine capable of the relevant performance level. For custom hardware, $c$ includes capital expenditure amortized over lifetime plus operational expenses, including energy, cooling, and maintenance.

\paragraph{Reward $V$.}
Estimating $V$ is more protocol-specific. In blockchain contexts, $V$ may include:
\begin{itemize}[leftmargin=*]
  \item Expected MEV attributable to early knowledge of beacon outputs.
  \item Increased probability of being selected as validator or committee member.
  \item Gains from side bets or derivatives conditioned on beacon outcomes.
\end{itemize}
Empirical MEV estimates from block explorers and mempool data can provide rough distributions and upper bounds for $V$. In many designs, $V$ is highly skewed: most rounds have low reward, but occasional spikes offer large gains. Our framework accommodates these distributions via $\E[V]$ and tail bounds.

\subsection{Illustrative Case Studies}

We now instantiate our model in three representative settings to illustrate how the theoretical conditions translate into concrete parameter choices.  
Throughout, we interpret $T$ as a wall-clock delay in seconds and $c$ as a cost in USD per second of adversarial running time.  
Unless stated otherwise, we use the baseline parameters
\[
  \delta = 3.0,
  \qquad
  c = 0.05 \text{ USD/s},
\]
which roughly correspond to a factor-$3$ hardware speedup and moderate cloud rental prices.

\paragraph{Case Study 1: VDF Beacon with 2--5 Second Delay.}

Consider a blockchain protocol that proposes a VDF-based beacon with delay $T \in [2,5]$ seconds, evaluated on commodity CPUs.  
Suppose that:

\begin{itemize}[leftmargin=*]
  \item high-end accelerators (FPGAs/ASICs) achieve an effective speedup of $\delta = 3$ over the honest baseline;
  \item the adversary rents such hardware at a rate of $c = 0.05$\,USD per second of adversarial running time;
  \item the expected MEV per round is $\E[V] = 10$\,USD in typical conditions, with spikes to $50$--$100$\,USD during congestion \cite{judmayer2022estimating}.
\end{itemize}

In the simplified single-round model, the expected profit from attacking in a round with reward $V$ is
\[
  \E[\profit(T)] \;=\; V - \frac{cT}{\delta}.
\]
For the parameters above,
\[
  \frac{c}{\delta} = \frac{0.05}{3} \approx 0.0167,
\]
so the expected cost per round is approximately $0.0167 T$\,USD.  
The break-even delay $T^\star$ that makes $\E[\profit(T^\star)] = 0$ is
\[
  T^\star \;=\; \frac{\delta}{c}\,\E[V] \;\approx\; 60\,\E[V],
\]
with $T^\star$ in seconds and $\E[V]$ in USD. For the three reward levels above:
\begin{align*}
  \E[V]=10\text{ USD} &: \quad T^\star \approx 600\text{ s } (\approx 10\text{ min}),\\
  \E[V]=50\text{ USD} &: \quad T^\star \approx 3{,}000\text{ s } (\approx 50\text{ min}),\\
  \E[V]=100\text{ USD} &: \quad T^\star \approx 6{,}000\text{ s } (\approx 100\text{ min}).
\end{align*}

\begin{figure}[htbp]
  \centering
  \includegraphics[width=0.8\linewidth]{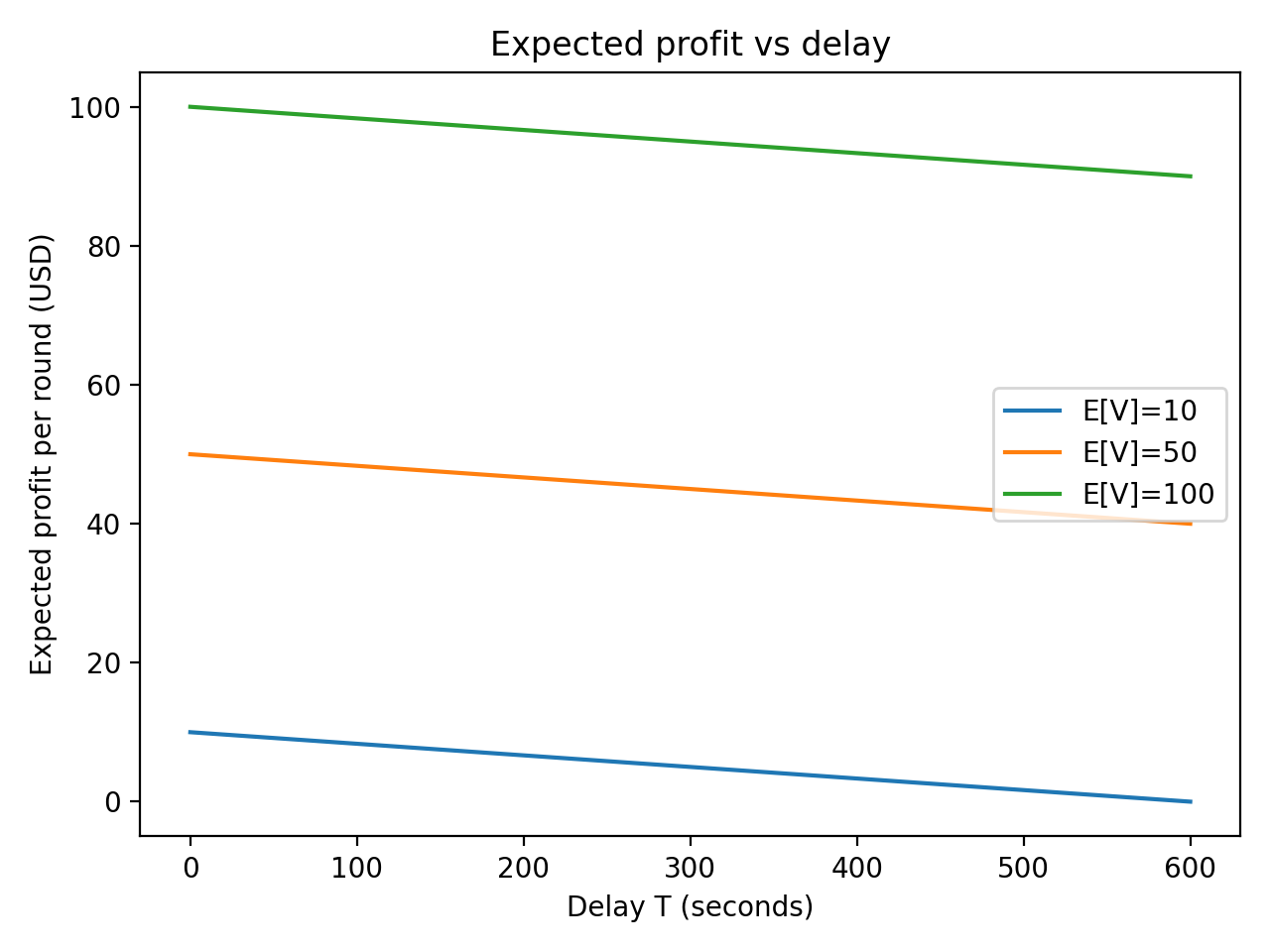}
  \caption{Expected profit per attack as a function of delay $T$ for different reward levels, assuming $c = 0.05$\,USD/s and $\delta = 3$. The economically secure region lies below the horizontal axis. Delays of a few seconds are far from sufficient when $\E[V]$ reaches tens of USD.}
  \label{fig:profit-vs-delay}
\end{figure}

Figure~\ref{fig:profit-vs-delay} plots $\E[\profit(T)]$ as a function of $T$ for these three reward levels.  
Delays of 2--5 seconds lie at the far left of the plot, deep in the region where attacks remain strongly profitable whenever $\E[V] \gtrsim 10$\,USD.

This case study shows that once MEV reaches even modest levels, economically secure delays are on the order of minutes rather than seconds, unless hardware is significantly more expensive or protocol design sharply limits MEV.

\paragraph{Case Study 2: Public Randomness Service.}

Next, consider a public randomness service, for instance, a consortium-operated beacon, that emits one output every $\Delta$ seconds using a VDF running on dedicated hardware.  
Assume:

\begin{itemize}[leftmargin=*]
  \item the beacon uses a VDF whose delay parameter is set to $T = \Delta$;
  \item the operator places a firm upper bound $V_{\max}$ on the economic value at stake in each draw by limiting stakes or prize size;
  \item a successful manipulation yields at most $V_{\max}$ in benefit to an adversary.
\end{itemize}

If the attacker can access hardware with speedup $\delta = 3$ and cost rate $c = 0.05$\,USD/s, then defending against the worst-case reward $V_{\max}$ yields the single-round threshold
\[
  \Delta^\star \;\ge\; T^\star \;=\; \frac{\delta}{c}\,V_{\max} \;=\; 60\,V_{\max}.
\]
Figure~\ref{fig:delta-vs-vmax} plots the required delay $\Delta^\star$ as a function of $V_{\max}$ under these parameters.

\begin{figure}[htbp]
  \centering
  \includegraphics[width=0.8\linewidth]{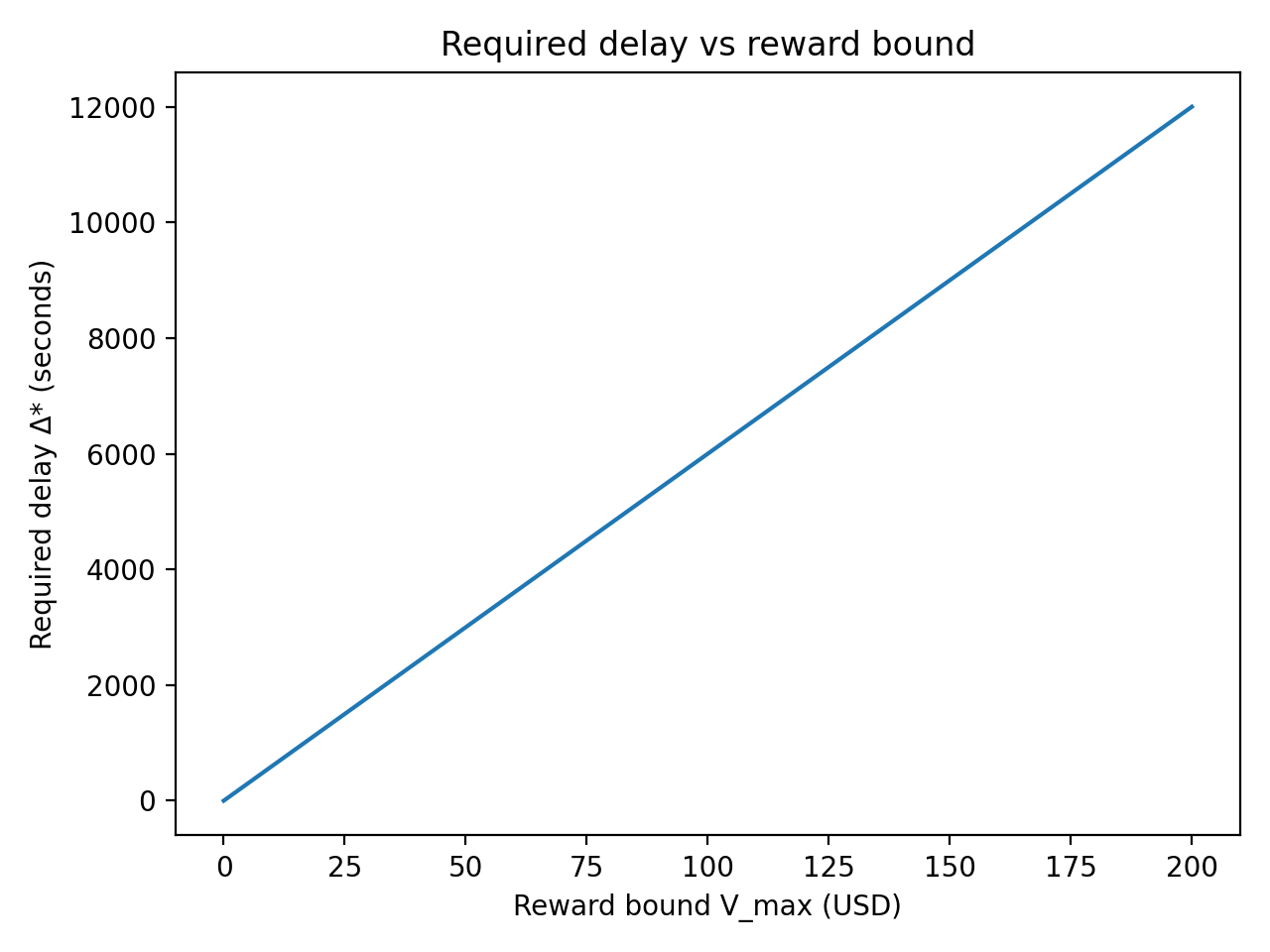}
  \caption{Required delay $\Delta^\star$ as a function of the bound on adversarial reward $V_{\max}$, with $\delta = 3$ and $c = 0.05$\,USD/s. For example, $V_{\max}=100$\,USD requires $\Delta^\star \approx 6{,}000$\,s (about 100 minutes) to be economically secure under this hardware model.}
  \label{fig:delta-vs-vmax}
\end{figure}

This reveals an inherent tradeoff in system design. For a given hardware model, economic security can be achieved only by setting $\Delta$ to a comparatively long duration, often tens of minutes for moderate $V_{\max}$, or by imposing firm constraints on the maximum economic value per round. Limiting the value at risk enables significantly shorter delays.

\paragraph{Case Study 3: Grinding in Committee Selection.}

Finally, we examine grinding in protocols where proposers can influence the seed used by the beacon, for example by selecting among multiple transaction orderings or candidate blocks.  
Suppose an adversary can explore $G$ candidate seeds $s_1,\ldots,s_G$ and evaluate the VDF for each, then choose the most favorable outcome.

Let $V^{(i)}$ denote the reward associated with seed $s_i$, and define
\[
  V_{\max} = \max_{1 \le i \le G} V^{(i)}.
\]
For illustration, assume that:

\begin{itemize}[leftmargin=*]
  \item the rewards $V^{(i)}$ associated with different seeds are independent and identically distributed with an exponential distribution of mean $\mu = 10$,USD. Under this distribution, the expected maximum over $G$ trials satisfies
    \[
      \E[V_{\max}] = \mu \, H_G,
    \]
    where $H_G$ is the $G$th harmonic number;
  \item the adversary can deploy $G$ partially shared computation streams, leading to an effective cost that scales as $c_{\mathrm{eff}}(G) = c\,G^{1/2}$ rather than increasing linearly with $G$.
  \item the computational speedup available to the adversary remains fixed at $\delta = 3$.
\end{itemize}

In this toy model, Theorem~\ref{thm:grinding-extended} suggests a required delay
\[
  T^\star_{\mathrm{grind}}(G)
  \;\approx\;
  \frac{\delta}{c_{\mathrm{eff}}(G)}\,\E[V_{\max}]
  \;=\;
  \frac{\delta \mu}{c}\,\frac{H_G}{G^{1/2}}.
\]

Figure~\ref{fig:grinding-delay} plots $T^\star_{\mathrm{grind}}(G)$ for $G$ up to $2^{10}$ on a logarithmic $G$-axis.

\begin{figure}[htbp]
  \centering
  \includegraphics[width=0.8\linewidth]{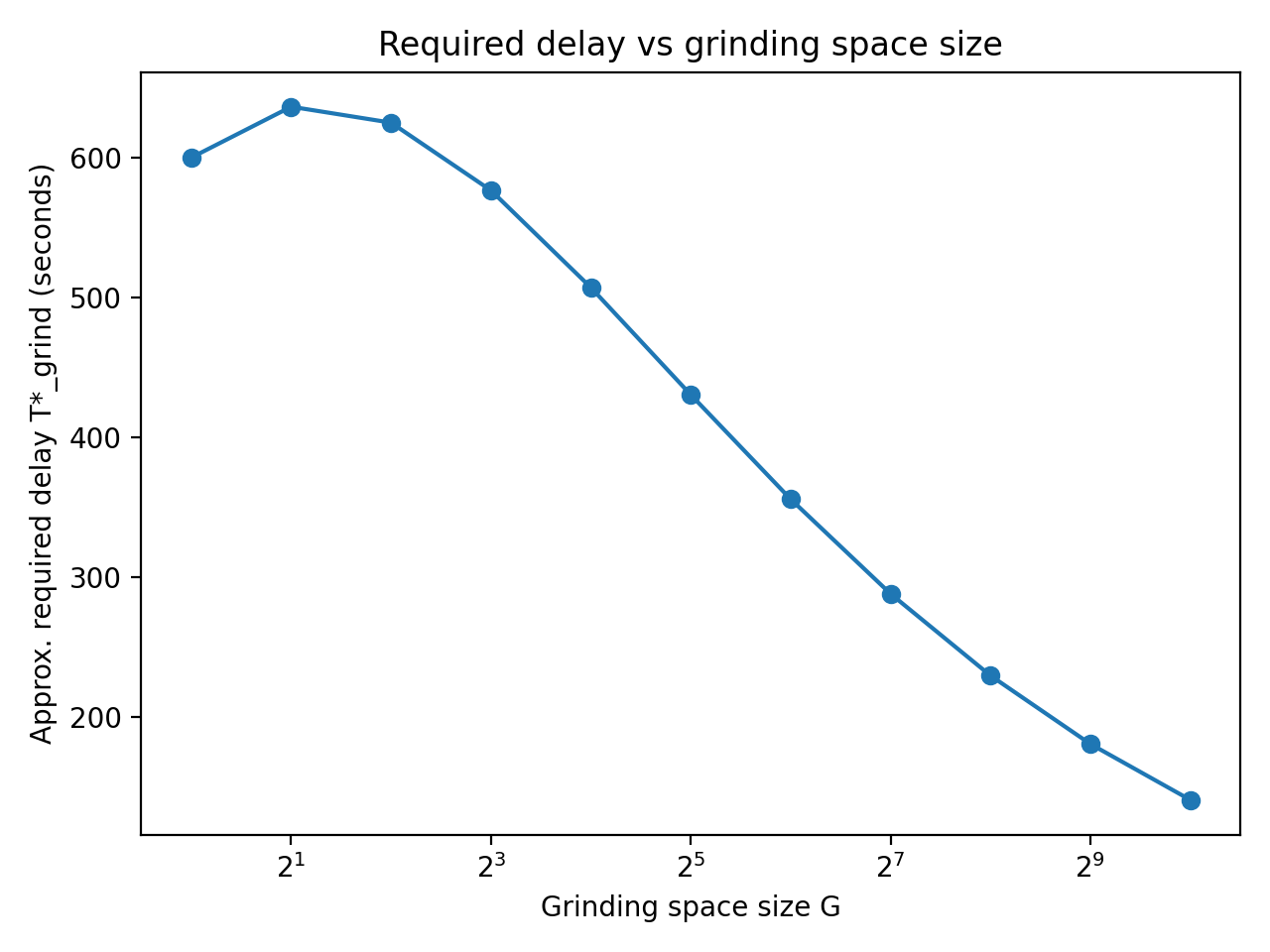}
  \caption{Illustrative required delay $T^\star_{\mathrm{grind}}(G)$ as a function of grinding space size $G$ on a log scale (base 2), assuming $\mu = 10$\,USD, $\delta = 3$, $c = 0.05$\,USD/s, and sublinear cost scaling $c_{\mathrm{eff}}(G) = c G^{1/2}$. Grinding increases the effective reward via $V_{\max}$, and depending on hardware scaling, may still require substantially larger delays.}
  \label{fig:grinding-delay}
\end{figure}

Although this example is stylized, it illustrates a robust qualitative point:
if grinding opportunities are not tightly constrained, even moderate $G$ can significantly amplify the effective economic value of manipulating the beacon, forcing protocol designers either to increase $T$ or to reduce $G$ by changing seed-derivation and leader-selection rules.

\paragraph{Case Study 4: Ethereum-Style Validator Selection with RANDAO.}
We consider a concrete scenario inspired by Ethereum's beacon chain. In each slot (12 seconds), a block proposer is selected pseudo-randomly using RANDAO. The proposer can extract MEV from transaction ordering.

According to Flashbots data, median MEV per block was approximately \$50 in 2023, with 99th-percentile values exceeding \$10,000 during periods of high volatility. Suppose a VDF-based beacon replaces RANDAO to prevent last-revealer manipulation. We ask: what delay $T$ is required for economic security?

Using current FPGA benchmarks for repeated squaring in RSA groups, a state-of-the-art FPGA achieves approximately $\delta = 2.5\times$ speedup over optimized CPU implementations. Cloud FPGA rental (e.g., AWS F1 instances) costs approximately \$1.65/hour $\approx$ \$0.00046/second.

For median MEV (\$50): $T^* = \delta \cdot \E[V] / c = 2.5 \times 50 / 0.00046 \approx 271{,}739$ seconds $\approx 3.1$ days. For 99th-percentile MEV (\$10,000): $T^* \approx 54{,}347{,}826$ seconds $\approx 629$ days.

These numbers are clearly impractical for a 12-second slot time, confirming that VDF-based beacons alone cannot provide economic security against MEV-motivated manipulation without additional protocol-level defenses (e.g., encrypted mempools, MEV redistribution, or threshold VDF schemes that distribute computation among multiple parties).

\section{Design Guidelines and ESDP}
\label{sec:design}

We summarize our findings as practical design guidelines and introduce an abstraction for protocol specifications.

\subsection{Design Guidelines}

\paragraph{Step 1: Model the reward distribution.}
Identify all channels through which an adversary could profit by manipulating or learning beacon outputs early. Estimate $\E[V]$ and, where necessary, tail bounds for $V$ and for $V_\text{max}$ in grinding scenarios.

\paragraph{Step 2: Estimate adversarial capabilities.}
Determine plausible ranges for $\delta$ based on performance benchmarks, and derive $c$ from cloud-pricing data or anticipated hardware costs. Consider adversaries that pool resources or rent specialized hardware dynamically.

\paragraph{Step 3: Choose $T$ to satisfy economic security.}
Apply Theorem~\ref{cor:linear} and its extensions to derive required lower bounds on $T$:
\[
  T \ge \max \left\{
    \frac{\delta}{c} \E[V], \;
    \frac{\delta}{cG} \E[V_\text{max}], \;
    \frac{\delta}{c} \beta(p) \E[V], \;
    \dots
  \right\}.
\]

\paragraph{Step 4: Account for multi-round effects.}
If the protocol's security depends on sequences of beacon outputs, incorporate multi-round constraints such as~\eqref{thm:multiround} into parameter selection.

\paragraph{Step 5: Revisit parameters periodically.}
Market conditions change. Cloud prices, hardware efficiency, and MEV volatility evolve over time. Designers should treat $T$ as a dynamically adjustable parameter that is periodically recalibrated using up-to-date costs and rewards.

\paragraph{Cost to Honest Nodes.} Our optimization for $T^*$ focuses on the adversary's cost-benefit analysis. However, the delay parameter also imposes costs on honest participants, who must wait $T$ seconds per round before the beacon output is available. In latency-sensitive applications such as block production, longer delays reduce throughput and increase confirmation times. The optimal delay from a system-design perspective must therefore balance economic security (requiring large $T$) against usability (requiring small $T$). Formally, the designer solves $\min_T \mathcal{L}(T)$ subject to $T \geq T^*$, where $\mathcal{L}(T)$ captures the system-level cost of delay (e.g., reduced throughput, increased finality time). This framing makes explicit that economic security is one constraint among several in practical parameter selection.

\subsection{Economically Secure Delay Parameters (ESDP)}

To facilitate adoption, we propose the following abstraction:

\begin{definition}[Economically Secure Delay Parameters (ESDP)]
An \emph{Economically Secure Delay Parameter} for a VDF-based beacon is a delay value $T^*$ such that, given specified bounds on adversarial speedup $\delta$, cost parameter $c$, and reward distribution (possibly including grinding and abort effects), the beacon is economically secure for all $T \ge T^*$.
\end{definition}

Protocol specifications can declare ESDPs explicitly, e.g.:

\begin{quote}
  \emph{For the expected range of MEV and hardware costs, and assuming adversarial speedup $\delta \le 4$, the delay parameter $T$ must satisfy $T \ge 8 \mathrm{s}$ to maintain economic security.}
\end{quote}

This enables designers and auditors to reason about security in a transparent, economically grounded manner and to adjust parameters as conditions evolve.

\section{Related Work}
\label{sec:related}

\paragraph{Verifiable Delay Functions.}
VDFs \cite{boneh2018verifiable} were introduced to formalize sequential work with succinct verification. Prior work has focused on constructing efficient VDFs based on repeated squaring in groups of unknown order, isogenies, and other number-theoretic assumptions \cite{ephraim2020continuous}\cite{wesolowski2020efficient}\cite{zhu2022low}; optimizing implementations \cite{mahmoody2019can}\cite{song2020high}; and integrating VDFs into systems such as randomness beacons and blockchains \cite{orlicki2020fair}\cite{venugopalan2023always}.

\paragraph{Randomness beacons.}
Randomness beacons \cite{choi2023sok}\cite{galindo2021fully} have a long history, from centralized sources to distributed protocols based on threshold signatures \cite{cascudo2023mt}, VRFs, and commit-reveal schemes \cite{choi2023bicorn}\cite{lee2025commit}. VDF-based beacons aim to improve bias-resistance and fairness in the presence of adaptive adversaries. Our work is orthogonal to cryptographic security of beacon constructions: we assume correctness and soundness of the underlying scheme and focus on economic incentives.

\paragraph{Rational cryptography and cryptoeconomics.}
Rational cryptography studies protocols where parties are modeled as economically rational agents \cite{caballero2007rational}\cite{garay2013rational}. In blockchain research, cryptoeconomic analyses often quantify incentives for consensus participation, selfish mining, and MEV extraction \cite{ganesh2024secure}\cite{huang2025optimizing}. Our contribution adapts these ideas to the specific setting of VDF-based randomness beacons, highlighting the need to configure parameters with economic considerations in mind.

\paragraph{MEV and protocol-level incentives.}
The literature on Maximal Extractable Value documents the significant value that can be extracted by reordering, including, or excluding transactions \cite{gramlich2024maximal}\cite{zhao2025mitigating}. This value directly feeds into the reward parameter $V$ in our model. Our work suggests that parameter choices for VDF-based beacons must consider MEV dynamics to remain economically secure.

\paragraph{Optimal RANDAO manipulation.}
Alpturer and Weinberg~\cite{alpturer2024optimal} study optimal RANDAO manipulation in Ethereum, concluding that manipulation bias is minimal under Ethereum's current parameters. Their analysis focuses on the combinatorial structure of RANDAO's XOR-based mixing and finds that the proposer's influence is bounded. Our work complements theirs by studying the economic incentives of VDF-based alternatives to RANDAO. While their result suggests RANDAO is relatively robust to bias in the short term, it does not address the economic security of VDF-based replacements, which face different attack surfaces (hardware speedup rather than last-revealer withholding). Together, the two analyses provide a more complete picture of randomness beacon security in blockchain systems.

\section{Discussion and Limitations}

Our framework relies on several modeling assumptions and simplifications. We briefly discuss key limitations.

\paragraph{Modeling $\delta$ and $c$.}
Estimating adversarial speedup and cost is inherently uncertain. Future hardware advances or economies of scale may significantly shift these parameters. We therefore recommend conservative estimates and periodic reevaluation.

\paragraph{Reward modeling.}
The reward distribution $V$ is protocol dependent and often heavy-tailed. While our theorems use $\E[V]$ and its variants, extreme-tail events may still be relevant for risk-averse designers. Extending the framework to explicitly incorporate risk preferences and worst-case guarantees is an interesting direction.

\paragraph{Dynamic strategies and complex environments.}
We focused on relatively simple attack strategies and states to obtain tractable analytic conditions. Real adversaries may use more complex strategies, including dynamic switching between attack modes, collusion, and integration with other economic activities. Our model can be extended to multi-agent settings and richer strategy spaces, but we leave detailed game-theoretic analysis to future work.

\paragraph{Cost Modeling and Amortization.}
We acknowledge that modeling cost as proportional to running time is a simplification. In practice, adversaries with purchased hardware face fixed capital costs amortized over many rounds, plus variable operating costs. Our framework accommodates this by setting $c$ to reflect the total amortized cost per unit of computation time. When capital costs dominate, $c$ decreases and the required delay $T^*$ increases accordingly, which is the conservative (safe) direction for protocol design.

\paragraph{Broader Beacon Designs.} Our analysis focuses on the simplest VDF beacon construction: a single VDF evaluated on a public seed. The design space for VDF-based beacons is considerably richer and includes threshold VDF schemes (where computation is distributed among multiple parties), chained constructions (where each round's output seeds the next), and hybrid designs combining VDFs with commit-reveal or verifiable random functions. While our economic framework does not directly apply to all such designs, the core insight---that delay parameters must be calibrated against economic incentives, not just cryptographic hardness---remains valid across the design space. Extending our analysis to threshold VDFs and chained constructions is an important direction for future work.

\paragraph{Beyond VDFs.}
Although our analysis focuses on VDF-based randomness beacons, the underlying economic perspective extends to a much broader range of randomness-generation mechanisms and cryptographic primitives. Any primitive whose security depends on parameter choices that influence adversarial cost is subject to the same fundamental tradeoffs between delay, hardware advantages, and economic incentives.

\section{Conclusion}
\label{sec:conclusion}

Verifiable Delay Functions provide a compelling foundation for secure randomness beacons in blockchains and other distributed systems. Yet cryptographic soundness alone does not guarantee safe deployment. A beacon is secure only if manipulating or prematurely learning its output is economically unprofitable for any rational adversary.

This work introduces a formal framework for reasoning about the economic security of VDF-based randomness beacons. We develop rational-adversary models, prove tight necessary and sufficient conditions linking delay parameters to hardware costs, adversarial speedup, and reward distributions, and extend the analysis to grinding attacks, selective abort, and multi-round settings. The resulting conditions yield actionable guidance for protocol designers and highlight the importance of aligning cryptographic guarantees with realistic economic environments.

We hope that the notion of \emph{Economically Secure Delay Parameters} and the methodology outlined in this work will become standard tools for the design and analysis of VDF-based beacons and related cryptographic mechanisms.



\bibliographystyle{ACM-Reference-Format}
\bibliography{refs} 

\appendix
\section*{Open Science Appendix}
\label{sec:openscience}

This paper follows the ACM CCS Open Science policy by documenting
all artifacts required to evaluate our results. All artifacts will be
provided to the program committee as an anonymized bundle via the
supplementary-material mechanism of the submission system. The
bundle contains no author-identifying metadata.

\subsection*{A. Artifacts Provided}

\paragraph{A.1 Jupyter notebook for numerical evaluation.}
We provide a single Python Jupyter notebook,
\texttt{vdf\_economic\_security.ipynb}, that reproduces all numerical
results and plots in Section~7. In particular, the notebook:
\begin{itemize}[leftmargin=*]
  \item implements the simplified single-round model
        $\mathbb{E}[\mathrm{profit}(T)] = \mathbb{E}[V] - (c/\delta)\,T$;
  \item computes and visualizes expected profit as a function of delay $T$
        for different reward levels $\mathbb{E}[V]$;
  \item computes and plots the required delay $T^\star = (\delta/c)\,V_{\max}$
        as a function of an upper bound $V_{\max}$ on per-round reward;
  \item implements a stylized grinding model with grinding space size $G$,
        harmonic expectation $\mathbb{E}[V_{\max}] = \mu H_G$, and
        sublinear cost scaling $c_{\mathrm{eff}}(G) = c\,G^{1/2}$, and
        plots the corresponding delay thresholds.
\end{itemize}
All figures in Section~7 are generated directly from this notebook.

\paragraph{A.2 Environment and dependencies.}
The notebook depends only on standard Python packages:
\texttt{numpy} and \texttt{matplotlib}. We include a short
\texttt{requirements.txt} specifying these dependencies. No
specialized hardware or external services are required.

\paragraph{A.3 Documentation.}
A short \texttt{README.md} file in the artifact bundle describes:
\begin{itemize}[leftmargin=*]
  \item how to open and run \texttt{vdf\_economic\_security.ipynb};
  \item how each figure in Section~7 is produced from specific cells;
  \item how to modify parameters $(\delta, c, \mathbb{E}[V], G)$ to explore
        alternative economic settings.
\end{itemize}

\subsection*{B. Artifacts Not Shared and Justifications}

\paragraph{B.1 Proprietary or non-public MEV traces.}
The motivation for our parameter choices references published MEV
estimates and hardware benchmarks from the literature and public
dashboards. We do not redistribute any proprietary raw MEV traces
or non-public datasets. Instead, the notebook uses simple parametric
models and synthetic values (for example, exponential rewards with
mean $\mu$ and bounded ranges for $V$) that are sufficient to
reproduce all figures and to verify the qualitative and quantitative
claims in the paper.

\paragraph{B.2 System-specific deployment data.}
We do not include logs, configuration files, or telemetry from any
production blockchain deployments. Such data may be subject to
privacy, contractual, or operational constraints. Our evaluation
relies only on abstracted parameter ranges and synthetic draws that
do not depend on deployment-specific details.

\subsection*{C. Access for Double-Blind Review}

All artifacts are accessible through an anonymous URL hosted on an independent file-sharing service:

https://anonymous.4open.science/r/VDF-Code-4343/

\subsection*{D. Reproducibility Statement}

Running the Jupyter notebook \texttt{vdf\_economic\_security.ipynb} in
a standard Python environment is sufficient to reproduce all
numerical results and plots in this paper. The theoretical
contributions (definitions, theorems, and proofs) are independent of
the artifacts and can be validated from the text alone.

\end{document}